\documentclass[prodmode]{notacm} 
\pdfoutput=1

\usepackage[ruled]{algorithm2e}

\SetAlFnt{\small}
\SetAlCapFnt{\small}
\SetAlCapNameFnt{\small}
\SetAlCapHSkip{0pt}
\IncMargin{-\parindent}

\usepackage[english]{babel}
\usepackage{amsmath}
\usepackage{amssymb}
\usepackage{latexsym}
\usepackage{amstext}
\usepackage{stmaryrd}
\usepackage{url}
\usepackage{color}
\usepackage{enumerate}
\usepackage{wrapfig}

\acmVolume{V}
\acmNumber{N}
\acmArticle{A}
\acmYear{YYYY}
\acmMonth{0}

\newcommand{\setsorts}{\mathcal{S}}
\newcommand{\setvars}{\mathcal{V}}

\newcommand{\arrtype}{\Rightarrow}
\newcommand{\arrfunc}{\Longrightarrow}
\newcommand{\T}{\mathcal{I}}
\newcommand{\Terms}{\mathcal{T}\!\!\mathit{erms}}
\newcommand{\J}{\mathcal{J}}

\newcommand{\interpret}[1]{\llbracket #1 \rrbracket}
\newcommand{\Bool}{\mathbb{B}}

\newcommand{\Constructors}{\mathcal{C}\mathit{ons}}
\newcommand{\Defineds}{\mathcal{D}}

\newcommand{\arrz}{\to}
\newcommand{\arr}[1]{\arrz_{#1}}

\newcommand{\constraint}[1]{[#1]}

\newcommand{\FV}{\mathit{Var}}
\newcommand{\domain}{\mathit{Dom}}
\newcommand{\Rules}{\mathcal{R}}
\newcommand{\Rulescalc}{\mathcal{R}_{\mathtt{calc}}}
\newcommand{\subst}[2]{#1#2}

\newcommand{\Sigmaterms}{\Sigma_{\mathit{terms}}}
\newcommand{\Sigmalogic}{\Sigma_{\mathit{theory}}}
\newcommand{\Sigmatheory}{\Sigma_{\mathit{theory}}}

\newcommand{\Sigmaint}{\Sigma_{\mathit{theory}}^{\mathit{int}}}

\newcommand{\Values}{\mathcal{V}al}

\newcommand{\LVar}{\mathit{LVar}}

\newcommand{\seq}[1]{\overrightarrow{\!#1}}

\newcommand{\afun}{f}

\newcommand{\avar}{x}
\newcommand{\bvar}{y}
\newcommand{\cvar}{z}

\newcommand{\aterm}{s}
\newcommand{\bterm}{t}
\newcommand{\cterm}{u}
\newcommand{\dterm}{w}
\newcommand{\eterm}{q}
\newcommand{\asort}{\iota}
\newcommand{\bsort}{\kappa}
\newcommand{\csort}{\mu}

\newcommand{\symb}[1]{\mathsf{#1}}

\newcommand{\Z}{\mathbb{Z}}

\newcommand{\nul}{\symb{0}}
\newcommand{\one}{\symb{1}}

\newcommand{\summ}{\symb{sum}}
\newcommand{\fact}{\symb{fact}}

\newcommand{\strue}{\symb{true}}
\newcommand{\sfalse}{\symb{false}}

\newcommand{\Int}{\symb{int}}
\newcommand{\IArr}{\symb{array(int)}}
\newcommand{\BOOL}{\symb{bool}}

\newcommand{\quant}[3]{#1#2 (#3)}

\newcommand{\exshort}{Example~}

\newcommand{\tool}[1]{\textsf{#1}}

\newcommand{\ctrl}{\tool{Ctrl}}

\newcommand{\OK}{\mbox{\textbf{\textit{OK}}}}

\newcommand{\myparagraph}[1]{\medskip\noindent\textit{#1.}}
\definecolor{comment}{rgb}{0.92, 0.92, 0.92}
\newcommand{\commentbox}[1]{\colorbox{comment}{\parbox{13cm}{
\emph{Comment:} #1}}}

\newcommand{\olddiscussion}[2]{} 

\definecolor{DarkGreen}{rgb}{0.0,0.50,0.0}

\begin{document}

\markboth{C. Kop}{Quasi-reductivity of LCTRSs}

\title{Quasi-reductivity of Logically Constrained Term Rewriting Systems}
\author{CYNTHIA KOP
\affil{University of Innsbruck and University of Copenhagen}
\vspace{-12pt}
}

\begin{abstract}
This paper considers \emph{quasi-reductivity}---essentially, the
property that an evaluation cannot get ``stuck'' due to a missing
case in pattern matching---in the context of term rewriting with
logical constraints.
\end{abstract}

\begin{bottomstuff}
This work is supported by the Austrian Science Fund (FWF)
international project I963, and the Marie-Sk{\l}odowska-Curie action
``HORIP'', program H2020-MSCA-IF-2014, 658162.
\end{bottomstuff}

\maketitle

\section{Introduction}
\label{sec:intro}

The formal framework of \emph{Logically Constrained Term Rewriting
Systems} (LCTRSs), introduced in \cite{kop:nis:13}, combines term
rewriting with constraints and calculations over an arbitrary theory.
This for instance allows users to specify rules with integers, arrays
and strings, and can be used to analyze both imperative and
functional programs (without higher-order variables) in a natural way.

Many methods to analyze term rewriting systems naturally extend to
LCTRSs.  In this paper we will study \emph{quasi-reductivity}, the
property that the only ground irreducible terms are constructor terms.
We provide a simple method to prove quasi-reductivity: essentially,
we will test that the rules do not omit any patterns.

\medskip
\emph{Structure:} For completeness, we will first set out the
definition of LCTRSs, following \cite{kop:nis:13,kop:nis:14} in
Section~\ref{sec:prelim}.
In Section~\ref{sec:quasireductivity} we consider the definition of
quasi-reductivity; in Section~\ref{sec:restrictions} we present
three restrictions:
\emph{left-linearity}, \emph{constructor-soundness} and
\emph{left-value-freeness}.  The core of this work is
Section~\ref{sec:algorithm}, where
we provide an algorithm to confirm quasi-reductivity for LCTRSs which
satisfy these restrictions, and prove its soundness.

\section{Preliminaries}
\label{sec:prelim}

In this section, we briefly recall \emph{Logically Constrained Term
Rewriting Systems} (usually abbreviated as \emph{LCTRSs}),
following the definitions in~\cite{kop:nis:13}.

\myparagraph{Many-sorted terms}
We introduce terms, typing, substitutions, contexts and subterms
(with corresponding terminology) in the usual way for many-sorted
term rewriting.

\begin{definition}
We assume given a set $\setsorts$ of \emph{sorts} and an infinite set
$\setvars$ of \emph{variables}, each variable equipped with a sort.
A \emph{signature} $\Sigma$ is a set of \emph{function symbols}
$\afun$, disjoint from $\setvars$, each 
equipped with a
\emph{sort declaration} $[\asort_1 \times \cdots \times \asort_n]
\arrtype \bsort$, with all $\asort_i$ and $\bsort$ sorts.
For readability, we often write $\bsort$ instead of
$[] \arrtype \bsort$.
The set $\Terms(\Sigma,\setvars)$ of \emph{terms} over $\Sigma$ and
$\setvars$ contains any expression $\aterm$ such that $\vdash \aterm
: \asort$ can be derived for some sort $\asort$, using:

\vspace{-4pt}
\noindent
\begin{minipage}[b]{0.28\linewidth}
\[
\frac{}{\vdash \avar : \asort}
\ (\avar : \asort \in \setvars)
\]
\end{minipage}
\begin{minipage}[b]{0.71\linewidth}
\[
\frac{\vdash \aterm_1 : \asort_1\ \ \ldots\ \ \vdash
\aterm_n : \asort_n
}{\vdash \afun(\aterm_1,\ldots,\aterm_n) : \bsort}
\ (\afun : [\asort_1 \times \cdots \times \asort_n] \arrtype \bsort 
\in \Sigma)
\]
\end{minipage}
\end{definition}

We fix $\Sigma$ and $\setvars$.  Note that for every term $\aterm$,
there is a unique sort $\asort$ with $\vdash \aterm : \asort$.

\begin{definition}
Let $\vdash \aterm : \asort$.
We call $\asort$ the \emph{sort of} $\aterm$.
Let $\FV(\aterm)$ be the set of variables occurring in $\aterm$;
we say that $\aterm$ is \emph{ground} if $\FV(\aterm) = \emptyset$.
\end{definition}

\begin{definition}
A \emph{substitution} $\gamma$ is a sort-preserving total mapping
from $\setvars$ to $\Terms(\Sigma,\setvars)$.
The result $\subst{\aterm}{\gamma}$ of applying a substitution
$\gamma$ to a term $\aterm$ is $\aterm$ with all occurrences of
a variable $\avar$ replaced by $\gamma(\avar)$.
The \emph{domain} of this substitution, $\domain(\gamma)$, is the set
of variables $x$ with $\gamma(x) \neq x$.
The notation $[\avar_1:=\aterm_1,\ldots,\avar_k:=\aterm_k]$ denotes
a substitution $\gamma$ with $\gamma(\avar_i) = \aterm_i$ for
$1 \leq i \leq n$, and $\gamma(y) = y$ for $y \notin \{\avar_1,
\dots,\avar_n\}$.
\end{definition}

\begin{definition}
A \emph{context} $C$ is a term containing a typed \emph{hole} $\Box :
\asort$.
If $\bterm : \asort$, we define $C[\bterm_1]$ as $C$ with $\Box$
replaced by $\bterm$.
If we can write $\aterm = C[\bterm]$, then $\bterm$ is a
\emph{subterm} of $\aterm$.
\end{definition}

\myparagraph{Logical terms}
Specific to LCTRSs, we consider different kinds of symbols and terms.

\begin{definition}
We assume given:
\begin{itemize}
\item signatures $\Sigmaterms$ and $\Sigmalogic$ such that
  $\Sigma = \Sigmaterms \cup \Sigmalogic$;
\item a mapping $\T$ which assigns to each sort $\asort$ occurring in
  $\Sigmalogic$ a set $\T_\asort$;
\item a mapping $\J$ which assigns to each $\afun : [\asort_1 \times
  \cdots \times \asort_n] \arrtype \bsort \in \Sigmalogic$ a function
  in $\T_{\asort_1} \times \cdots \times \T_{\asort_n} \arrfunc
  \T_\bsort$;
\item for all sorts $\asort$ occurring in $\Sigmalogic$ a set
  $\Values_\asort \subseteq \Sigmalogic$ of \emph{values}: function
  symbols $a : [] \arrtype \asort$ such that $\J$ gives a bijective
  mapping from $\Values_\asort$ to $\T_\asort$.
\end{itemize}
We require that $\Sigmaterms \cap \Sigmalogic \subseteq \Values%
= \bigcup_\asort \Values_\asort$.
The sorts occurring in $\Sigmalogic$ are called \emph{theory sorts},
and the symbols \emph{theory symbols}.
Symbols in $\Sigmalogic \setminus \Values$ are \emph{calculation
  symbols}.
A term in $\Terms(\Sigmalogic,\setvars)$ is called a \emph{logical
term}.
\end{definition}

\begin{definition}
For ground logical terms, let
$
\interpret{\afun(\aterm_1,\ldots,\aterm_n)} :=
\J_\afun(\interpret{\aterm_1},\ldots,\interpret{\aterm_n})
$.
Every ground logical term $\aterm$ corresponds to a unique value $c$
such that $\interpret{\aterm} = \interpret{c}$; we say that $c$ is the
value of $\aterm$.
%
A \emph{constraint} is a logical term $\varphi$ of some sort $\BOOL$
with $\T_\BOOL = \Bool = \{ \top,\bot \}$, the set of
\emph{booleans}.
A constraint
$\varphi$ is \emph{valid} if
$\interpret{\subst{\varphi}{\gamma}} = \top$ for \emph{all}
substitutions
$\gamma$ which map $\FV(\varphi)$ to values,
and \emph{satisfiable} if 
$\interpret{\subst{\varphi}{\gamma}} = \top$ for \emph{some}
substitutions
$\gamma$ which map $\FV(\varphi)$ to values.
A substitution $\gamma$ \emph{respects}
$\varphi$ if
$\gamma(\avar)$ is a value for all $\avar \in \FV(\varphi)$ and
$\interpret{\varphi\gamma} = \top$.
\end{definition}

Formally, terms in $\Terms(\Sigmaterms,\setvars)$ have no special
function, but we see them as the primary objects of our term rewriting
systems: a reduction would typically begin and end with such terms,
with
calculation symbols
only used in intermediate terms.  Their
function is to perform calculations in the underlying theory.
Usually, values which are expected to occur in starting terms and
end terms should be included both in $\Sigmaterms$ and $\Sigmalogic$,
while values only used in constraints and calculations would only be
in $\Sigmalogic$; $\strue$ and $\sfalse$ often fall in the latter
category.

\begin{example}\label{exa:factsignature}
Let $\setsorts = \{ \Int,\BOOL \}$, and consider the signature
$\Sigma = \Sigmaterms \cup \Sigmalogic$ where
$
\Sigmaterms = \{\ \fact : [\Int] \arrtype \Int\ \} \cup \{\ \symb{n} : 
  \Int \mid n \in \Z\ \}
$
and
$\Sigmalogic = \{\ \strue,\sfalse : \BOOL,\wedge,\vee,\Rightarrow :
[\BOOL \times \BOOL] \arrtype \BOOL,\ +,-,* : [\Int \times \Int]
\arrtype \Int,\ \leq, <, = : [\Int \times \Int] \arrtype \BOOL\ 
\} \cup \{\symb{n} : \Int \mid n \in \Z\}$. 
Then both $\Int$ and $\BOOL$ are theory sorts, and the values are
$\strue$, $\sfalse$ and all symbols $\symb{n}$ representing integers.
For the interpretations, let $\T_\Int = \Z$,\ $\T_\BOOL = \Bool$, and
let $\J$ be the evaluation function which interprets these symbols as
expected.

Using infix notation, examples of logical terms are $\nul =
\nul+-\one$ and $\avar+\symb{3} \leq \bvar + -\symb{42}$.  Both are
constraints.  $\symb{5}+\symb{9}$ is also a (ground) logical term, but
not a constraint.
Expected starting terms are for instance $\fact(\symb{42})$ or
$\fact(\fact(\symb{-4}))$: ground terms fully built using symbols in
$\Sigmaterms$.
\end{example}

\myparagraph{Rules and rewriting}
We adapt the standard notions of rewriting (see, e.g.,
\cite{baa:nip:98}) by including constraints and adding rules to
perform calculations.

\begin{definition}
A \emph{rule} is a triple $\ell \arrz r\ \constraint{\varphi}$ where
$\ell$ and $r$ are terms of the same sort and $\varphi$ is a
constraint.  Here, $\ell$ has the form $f(\ell_1,\dots,\ell_n)$ and
contains at least one symbol in $\Sigmaterms \setminus \Sigmatheory$
(so $\ell$ is not a logical term).
If $\varphi = \symb{true}$ with $\J(\symb{true}) = \top$, the rule
may be denoted $\ell \arrz r$.
Let $\LVar(\ell \arrz r\ \constraint{\varphi})$ denote $\FV(\varphi)
\cup (\FV(r) \setminus \FV(\ell))$.
A substitution $\gamma$ \emph{respects} $\ell \arrz r\ 
\constraint{\varphi}$ if
$\gamma(\avar)$ is a value for all $\avar \in
\LVar(\ell \arrz r\ \constraint{\varphi})$, and 
$\interpret{\varphi\gamma} = \top$.
\end{definition}

Note that it is allowed to have $\FV(r) \not \subseteq \FV(\ell)$, but
fresh variables in the right-hand side may only be instantiated
with \emph{values}.  This is done to model user input or random
choice, both of which would typically produce a value.  Variables in
the left-hand sides do not need to be instantiated with values
(unless they also occur in the constraint); this is needed for
instance to support a lazy evaluation strategy.

\begin{definition}
We assume given a set of rules $\Rules$ and let
$\Rulescalc$ be the set
$
\{ \afun(\avar_1,\ldots,\avar_n) \arrz \bvar\ 
\constraint{\bvar = \afun(\seq{\avar})} \mid \afun : [\asort_1
\times \cdots \times \asort_n] \arrtype \bsort \in \Sigmalogic
\setminus \Values \}$
(writing $\seq{\avar}$ for $\avar_1,\ldots,\avar_n$).
The \emph{rewrite relation} $\arrz_{\Rules}$ is a binary relation
on terms, defined by:
\[
\begin{array}{rcll}
C[\ell\gamma] & \arr{\Rules} & C[r\gamma]\ &
\text{if}\ 
\ell \arrz r\ \constraint{\varphi} \in \Rules \cup \Rulescalc\ 
\text{and}\ 
\gamma\ \text{respects}\ \ell \arrz r\ \constraint{\varphi} \\
\end{array}
\]
Here, $C$ is an arbitrary context.
%
A reduction step with $\Rulescalc$ is called a \emph{calculation}.
A term is in \emph{normal form} if it cannot be reduced with
$\arr{\Rules}$.
\end{definition}

We will usually call the elements of $\Rulescalc$ rules---or
\emph{calculation rules}--even though their left-hand side is a
logical term.

\begin{definition}
For $\afun(\ell_1,\ldots,\ell_n) \arrz r\ \constraint{\varphi}
\in \Rules$ we call $\afun$ a \emph{defined symbol}; non-defined
elements of $\Sigmaterms$ and all values are \emph{constructors}.
Let $\Defineds$ be the set of all defined symbols,
and $\Constructors$ the set of constructors.
A term in $\Terms(\Constructors,\setvars)$ is a
\emph{constructor term}.
\end{definition}

Now we may define a \emph{logically constrained term rewriting
system} (LCTRS) as the abstract rewriting system $(\Terms(\Sigma,
\setvars),\arr{\Rules})$.  An LCTRS is usually given by
supplying $\Sigma$,\ $\Rules$, and also $\T$ and $\J$ if these are not
clear from context.

\begin{example}\label{exa:factlctrs}
To implement an LCTRS calculating the \emph{factorial} function, we
use the signature $\Sigma$ from \exshort\ref{exa:factsignature}, and
the following rules:
\[
\Rules_{\fact}
 = \{\ 
\fact(x) \to \one\ \constraint{x \leq \nul}\ \ ,\ \ 
\fact(x) \to x * \fact(x-\one)\ \constraint{\neg (x \leq \nul)}
\ \}
\]
Using calculation steps, a term $\symb{3}-\symb{1}$ reduces to
$\symb{2}$ in one step (using the calculation rule $\avar-
\bvar\arrz\cvar\ \constraint{\cvar = \avar-\bvar}$), and $\symb{3} *
(\symb{2} * (\symb{1} * \symb{1}))$ reduces to $\symb{6}$ in three
steps.  Using also the rules in
$\Rules_{\fact}$,
$\symb{fact}(\symb{3})$ reduces in ten steps to $\symb{6}$.
\end{example}

\begin{example}\label{exa:arraysumlctrs}
To implement an LCTRS calculating the sum of elements in an array,
let $\T_\BOOL = \Bool,\ \T_\Int = \Z,\ \T_\IArr = \Z^*$, so
$\IArr$ is mapped to finite-length integer sequences.
Let $\Sigmalogic = \Sigmaint \cup \{ \symb{size} : [\IArr] \arrtype
\Int,\ \symb{select} : [\IArr \times \Int] \arrtype \Int \}\ \cup\ 
\{ \symb{a} \mid a \in \Z^* \}$.
(We do \emph{not} encode arrays as lists: every ``array''---integer
sequence---$a$ corresponds to a unique symbol $\symb{a}$.)
The interpretation function $\J$ behaves on $\Sigmaint$ as usual,
maps the values $\symb{a}$ to the corresponding integer sequence, and
has:
\[
\begin{array}{rcll}
\J_{\symb{size}}(a) & = & k & \text{if}\ a = \langle n_0,\ldots,n_{k-1} \rangle \\
\J_{\symb{select}}(a,i) & = & n_i & \text{if}\ a = \langle n_0,\ldots, n_{k-1}
  \rangle\ \text{with}\ 0 \leq i < k \\
  & & 0 & \text{otherwise} \\
\end{array}
\]
In addition, let:
\[
\begin{array}{rcl}
\Sigmaterms & = & \{\ \summ : [\IArr] \arrtype \Int,\ \symb{sum0} :
  [\IArr \times \Int] \arrtype \Int\ \}\ \cup \\
  & & \{\ \symb{n} : \Int \mid n \in \Z\ \} \cup \{\ \symb{a} \mid a \in \Z^*\ \}
\end{array}
\]
\[
\Rules
 = \left\{
\begin{array}{rcll}
\summ(\avar) & \arrz & \symb{sum0}(\avar,\symb{size}(\avar)-\one) \\
\symb{sum0}(\avar,k) & \arrz & \symb{select}(\avar,k) +
  \symb{sum0}(\avar,k-\one) & \constraint{k \geq \nul} \\
\symb{sum0}(\avar,k) & \arrz & \nul & \constraint{k < \nul} \\
\end{array}\right\}
\]
\end{example}

Note the special role of \emph{values}, which are new in LCTRSs
compared to older styles of constrained rewriting.  Values are the
representatives of the underlying theory.  All values are constants
(constructor symbols $v()$ which do not take arguments),
even if they represent complex structures, as seen in
\exshort\ref{exa:arraysumlctrs}.  However, not all constants are values;
for instance a constant constructor $\symb{error} \in \Sigmaterms$
would not be a value.
We will often work with signatures having infinitely many values.
Note that we do not match modulo theories, e.g.~we do not equate
$\nul + (\avar + \bvar)$ with $\bvar + \avar$ for matching.

Note also the restriction on variables in a constraint being
instantiated by values; for instance in \exshort\ref{exa:factlctrs},
a term $\fact(\symb{fact}(\symb{3}))$ reduces only at the
\emph{inner} $\symb{fact}$.

%

\section{Quasi-reductivity}\label{sec:quasireductivity}

The most high-level definition of quasi-reductivity is likely the
following.

\begin{definition}[Quasi-reductivity]
An LCTRS $(\Sigmaterms,\Sigmalogic,\T,\J,\Rules)$ is quasi-reductive
if for all $\aterm \in \Terms(\Sigma,\emptyset)$ one of the following
holds:
\begin{itemize}
\item $\aterm \in \Terms(\Constructors,\emptyset)$ (we say: $\aterm$
  is a \emph{ground constructor term});
\item there is a $\bterm$ such that $\aterm \arr{\Rules} \bterm$ (we
  say: $\aterm$ \emph{reduces}).
\end{itemize}
\end{definition}

Note that $\Terms(\Sigma,\emptyset)$ is the set of \emph{ground}
terms.  Another common definition concerns only the reduction of
``basic'' ground terms, but is equivalent:

\begin{lemma}\label{lem:quasirephrase}\label{lem:minex}
An LCTRS is quasi-reductive if and only if all terms $\afun(\aterm_1,
\ldots,\aterm_n)$ with $\afun$ a defined or calculation symbol and
all $\aterm_i \in \Terms(\Constructors,\emptyset)$, reduce.
\end{lemma}

\begin{proof}
If the LCTRS is quasi-reductive, then each such $\afun(\seq{\aterm})$
reduces, as it is not a constructor term.  If the LCTRS is not
quasi-reductive, then let $\afun(\seq{\aterm})$ be a minimal ground
irreducible non-constructor term.  By minimality, all $\aterm_i$
must be constructor terms.  If $\afun$ is a constructor, then the
whole term is a constructor term, contradiction, so $\afun$ is either
a defined symbol or a calculation symbol.
\qed
\end{proof}

\section{Restrictions}\label{sec:restrictions}

For our algorithm in the next section, which proves that a given LCTRS
is quasi-reductive, we will limit interest to LCTRSs which satisfy the
following restrictions:

\begin{definition}[Restrictions]\label{def:restrictions}
An LCTRS $(\Sigmaterms,\Sigmalogic,\T,\J,\Rules)$ is:
\begin{itemize}
\item \emph{left-linear} if for all rules $\ell \arrz r\ 
  \constraint{\varphi} \in \Rules$: every variable in $\ell$ occurs
  only once;
\item \emph{constructor-sound} if there are ground constructor terms
  for every sort $\asort$ such that some $\afun : [\ldots \times
  \asort \times \ldots] \arrtype \bsort \in \Defineds$ (so for every
  input sort of a defined symbol);
\item \emph{left-value-free} if the left-hand sides of rules do not
  contain any values.
\end{itemize}
\end{definition}

Note that any LCTRS can be turned into a left-value-free one, by
replacing a value $v$ by a fresh variable and adding a constraint
$\avar=v$ instead.  Constructor-soundness seems quite natural,
with a sort representing the set of ground constructor terms of that
sort.  Left-linearity is probably the greatest limitation; however,
note that non-left-linear systems impose \emph{syntactic} equality.
In a rule
\[
\symb{addtoset}(\avar,\symb{setof}(\avar,\mathit{rest})) \arrz
\symb{setof}(\avar,\mathit{rest})
\]
we can reduce $\symb{addtoset}(\symb{3} + \symb{4},\symb{setof}(\symb{3}+\symb{4},\aterm))$
immediately to $\symb{setof}(\symb{3}+\symb{4},\aterm)$.  However, we cannot reduce
$\symb{addtoset}(\symb{3} + \symb{4},\symb{setof}(\symb{4}+\symb{3},\aterm))$ with this rule.
There is also no syntactic way to check for \emph{inequality}.
Therefore, it seems like we could better formulate this rule and its
complement using constraints, or (if the sort of $\avar$ has
non-value constructors) by a structural check.  The rule above and its
complement could for instance become:
\[
\begin{array}{rcll}
\symb{addtoset}(\avar,\symb{setof}(\bvar,\mathit{rest})) & \arrz &
\symb{setof}(\bvar,\mathit{rest}) & \constraint{\avar = \bvar} \\
\symb{addtoset}(\avar,\symb{setof}(\bvar,\mathit{rest})) & \arrz &
\symb{setof}(\bvar,\symb{addtoset}(\avar,\mathit{rest})) &
\constraint{\avar \neq \bvar} \\
\end{array}
\]
In this light, left-linearity also seems like a very natural
restriction.
\vspace{4pt}

\commentbox{In~\cite{fal:kap:12} a similar method is introduced to
prove quasi-reductivity of a different style of constrained rewriting.
There, however, the systems are additionally restricted to be
\emph{value-safe}: the only constructors of sorts occurring in
$\Sigmalogic$ are values.  We drop this requirement here, because it
is not necessary in the definition of LCTRSs.}

\medskip
Constructor-soundness, arguably the most innocent of these
restrictions, allows us to limit interest to certain well-behaved
rules when proving quasi-reductivity:

\begin{theorem}\label{thm:constructorrules}
A constructor-sound LCTRS with rules $\Rules$ is quasi-reductive if
and only if the following conditions both hold:
\begin{itemize}
\item the same LCTRS restricted to constructor rules $\Rules' :=
  \{ \afun(\seq{\ell}) \arrz r\ \constraint{\varphi} \in \Rules \mid
  \quant{\forall}{i}{\ell_i \in \Terms(\Constructors,\setvars)} \}$ is
  quasi-reductive;
\item all constructor symbols with respect to $\Rules'$ are also
  constructors w.r.t.~$\Rules$.
\end{itemize}
\end{theorem}

\begin{proof}
Suppose $\Rules'$ is quasi-reductive, and constructor terms are the
same in either LCTRS.
Then also $\Rules$ is quasi-reductive, as anything which reduces
under $\Rules'$ also reduces under $\Rules$.
Alternatively, suppose $\Rules$ is quasi-reductive.

Towards a contradiction, suppose $\Rules'$ has constructor symbols
which are not constructors in $\Rules$; let $\afun$ be such a symbol.
As $\afun$ is a constructor for $\Rules'$, there are no rules
$\afun(\seq{l}) \arrz r\ \constraint{\varphi} \in \Rules \cup
\Rulescalc$ which match terms of the form $\afun(\seq{\aterm})$ with
all $\aterm_i \in \Terms(\Constructors,\emptyset)$.
Because $\afun$ is a defined symbol, such terms exist by
constructor-soundness.
As nothing matches $\afun(\seq{\aterm})$ itself, and its strict
subterms are constructor terms so cannot be reduced, this term
contradicts quasi-reductivity of $\Rules$!

For the first point, suppose towards a contradiction that $\Rules'$
is not quasi-reductive, yet $\Rules$ is, and the same terms are
constructor terms in either.  By Lemma~\ref{lem:minex} there is some
irreducible $\afun(\aterm_1,\ldots,\aterm_n)$ with all $\aterm_i$
constructor terms and $\afun$ not a constructor.  As the $\aterm_i$
are constructor terms, the rules in $\Rules \setminus \Rules'$ also
cannot match!  Thus, the term is also irreducible with
$\arr{\Rules}$, contradiction.
\qed
\end{proof}

\section{An algorithm to prove quasi-reductivity}\label{sec:algorithm}

We now present an algorithm to confirm quasi-reductivity of a given
LCTRS satisfying the restrictions from
Definition~\ref{def:restrictions}.
Following Theorem~\ref{thm:constructorrules}, we can---without loss of
generality---limit interest to \emph{constructor TRSs}, where the
immediate arguments in the left-hand sides of rules are all
constructor terms.

\myparagraph{Main Algorithm}
We assume given sequences $\asort_1,\ldots,\asort_n$ of theory sorts,
$\bsort_1,\ldots, 
\bsort_m$ of sorts, and $\avar_1,\ldots,\avar_n$ of
variables, with each $\avar_i : \asort_i \in \setvars$.
Moreover, we assume given a set $A$ of pairs $(\seq{s},\varphi)$.
Here, $\seq{s}$ is a sequence $s_1,\ldots,s_m$ of constructor terms
which do not contain values, such that $\vdash s_i : \bsort_i$,
and $\varphi$ is a logical constraint.  The $s_i$ have no
overlapping variables with each other or the $\avar_j$; that is, a
term $\afun(\avar_1,\ldots,\avar_n,s_1,\ldots,s_m)$ would be linear.
Variables in $\seq{\avar}$ and $\seq{s}$ may occur in $\varphi$,
however.

Now, for $b \in \{\mathsf{term},\mathsf{value},\mathsf{either}
\}$,\footnote{This parameter indicates what constructor instantiations
we should consider for $\aterm_1$.} define the function
$\OK(\seq{\avar},A,b)$ as follows; this construction is
well-defined by induction first on the number of function symbols
occurring in $A$, second by the number of variables occurring in $A$,
and third by the flag $b$ (with $\mathsf{either} > \mathsf{term},
\mathsf{value}$).
Only symbols in the terms $\aterm_i$ are counted for the first
induction hypothesis, so not those in the constraints.
\begin{itemize}
\item if $m = 0$:
  let $\{ y_1, \ldots, y_k \} = (\bigcup_{((),\varphi) \in A}
    \FV(\varphi)) \setminus \{ x_1, \ldots, x_n \}$;
  \begin{itemize}
  \item if $\quant{\exists}{y_1 \ldots y_k}{\bigvee_{((),\varphi) \in
    A} \varphi}$ is valid, then \textbf{true}
  \item else \textbf{false}
  \end{itemize}
  Note that if $A = \emptyset$, this returns \textbf{false}.
\item if $m > 0$ and $b = \mathsf{either}$, then consider
  $\bsort_1$.  If $\bsort_1$ does not occur in $\Sigmalogic$,
  the result is: 
      \[
      \OK(\seq{\avar},A,\mathsf{term})
      \]
  If $\bsort_1$ occurs in $\Sigmalogic$ and all
  constructors with output sort $\bsort_1$ are values, then the
  result is: 
      \[
      \OK(\seq{\avar},A,\mathsf{value})
      \]
  If $\bsort_1$ occurs in $\Sigmalogic$ but there are also non-value
  constructors of sort $\bsort_1$, then let $V := \{ (\seq{\aterm},
  \varphi) \in A \mid \aterm_1$ is a variable$\}$ and $T :=
  \{ (\seq{\aterm},\varphi) \in A \mid \aterm_1$ is not a variable
  in $\FV(\varphi) \}$.  Note that $V$ and $T$ overlap in cases where
  $\aterm_1$ is a variable not occurring in $\varphi$.
  The result of the function is: 
      \[
      \OK(\seq{\avar},V,\mathsf{value}) \wedge
      \OK(\seq{\avar},T,\mathsf{term})
      \]
  In all cases, the recursive calls are defined, by the decrease in
  the third argument (and in the last case possibly also in the first
  and second argument).
\item if $m > 0$ and $b = \mathsf{value}$, then we assume that
  $\bsort_1$ occurs in $\Sigmalogic$ and for all $(\seq{\aterm},
  \varphi) \in A$ the first term, $\aterm_1$, is a variable (if not,
  we might define the function result as \textbf{false}, but this
  cannot occur in the algorithm).
  Then let $\avar_{n+1}$ be a fresh variable of sort $\bsort_1$, and
  let $A' := \{ ((\aterm_2,\ldots,\aterm_m),\varphi[\aterm_1:=
  \avar_{n+1}]) \mid ((\aterm_1,\ldots,\aterm_n),\varphi) \in A \}$;
  the result is: 
      \[
      \OK((\avar_1,\ldots,\avar_{n+1}),A',\mathsf{either})
      \]
  Note that $A'$ has equally many function symbols as and fewer
  variables than $A$, and that we indeed have suitable sort
  sequences ($\asort_1,\ldots, \asort_n,\bsort_1)$ for the variables
  and $\bsort_2,\ldots,\bsort_m$ for the term sequences);
\item if $m > 0$ and $b = \mathsf{term}$ and for all $(\seq{\aterm},
  \varphi) \in A$ the first term, $\aterm_1$, is a variable, then we
  assume (like we did in the previous case) that never $\aterm_1 \in
  \varphi$, and let
  $A' := \{ ((\aterm_2,\ldots,\aterm_m),\varphi) \mid ((\aterm_1,
  \ldots,\aterm_m), \varphi) \in A \}$; the result is: 
      \[
      \OK(\seq{\avar},A',\mathsf{either})
      \]
  $A'$ has at most as many function symbols as and
  fewer variables than $A$.
\item if $m > 0$ and $b = \mathsf{term}$ and there is some
  $(\seq{\aterm},\varphi) \in A$ where $\aterm_1$ is not a variable,
  then let $f_1,\ldots,f_k$ be all non-value constructors with output
  sort $\bsort_1$ and let $A_1,\ldots,A_k$ be defined as follows:
  $A_i := \{ (\seq{\aterm},\varphi) \in A \mid \aterm_1$ is a variable
  or has the form $\afun_i(\seq{\bterm}) \}$.

  Now, for all $i$: if $f_i$ has sort declaration $[\csort_1 \times
  \cdots \times \csort_p] \arrtype \bsort_1$, then we consider the new
  sort sequence $\seq{\bsort}'$ with $\seq{\bsort'} = (\csort_1,\ldots,
  \csort_p,\bsort_2,\ldots,\bsort_m)$; for every $(\seq{\aterm},\varphi)
  \in A_i$ we define:
  \begin{itemize}
  \item if $\aterm_1 = \afun_i(\bterm_1,\ldots,\bterm_p)$, then
    $\seq{\cterm} := (\bterm_1,\ldots,\bterm_p,\aterm_2,\ldots,
    \aterm_m)$ and $\psi := \varphi$
  \item if $\aterm_1$ is a variable, then let $\bvar_1,\ldots,
    \bvar_p$ be fresh variables with sorts $\csort_1,\ldots,
    \csort_p$ respectively, and let $\seq{\cterm} := (\bvar_1,
    \ldots,\bvar_p,\aterm_2,\ldots,\aterm_m)$ and $\psi := \varphi[
    \aterm_1:=\afun_i(\bvar_1,\ldots,\bvar_p)]$.
  \end{itemize}
  We let $B_i$ be the set of the corresponding $(\seq{\cterm},\psi)$
  for all $(\seq{\aterm},\varphi) \in A_i$.  Now, if $A_i$
  contains any element where $\aterm_i$ is not a variable, then $B_i$
  contains fewer function symbols (not counting constraints), so also
  fewer than $A$ (as $A_i \subseteq A$).  If all $\aterm_i$ are
  variables, then $A_i$ is a strict subset of $A$, which misses at
  least one element where $\aterm_i$ contains a symbol, so also $B_i$
  has fewer symbols.  Either way, we are safe defining the result as:
      \[
      \OK(\seq{\avar},B_1,\mathsf{either}) \wedge
      \cdots \wedge \OK(\seq{\avar},B_k,\mathsf{either})
      \]
\end{itemize}

\medskip
Correctness of this algorithm is proved using the following technical
result.

\begin{lemma}\label{lem:algorithm}
For any suitable $n,m,\seq{\asort},\seq{\bsort},\seq{\avar},b$ and
$A$ such that \OK$(\seq{\avar},A,b) = \mathtt{true}$, we have,
for any sequence $(\aterm_1,\ldots,\aterm_n)$ of values and any
sequence $(\bterm_1,\ldots,\bterm_m)$ of ground constructor terms: if
one of the following conditions holds,
\begin{itemize}
\item $b = \mathsf{either}$ or $m = 0$
\item $b = \mathsf{value}$ and $\bterm_1$ is a value
\item $b = \mathsf{term}$ and $\bterm_1$ is not a value, and for all
  $((\cterm_1,\ldots,\cterm_n),\varphi) \in A$: $\cterm_1 \notin
  \FV(\varphi)$
\end{itemize}
then there is some $((\cterm_1,\ldots,\cterm_m),\varphi) \in A$
and a substitution $\gamma$ with $\gamma(\avar_i) = \aterm_i$ such
that:
\begin{itemize}
\item each $\bterm_i = \cterm_i\gamma$
\item $\varphi\gamma$ is a valid ground logical constraint
\end{itemize}
\end{lemma}

\begin{proof}
By induction on the derivation of
\OK$(\seq{\avar},A,b) = \mathtt{true}$.
Let values $(\aterm_1,\ldots,\aterm_n)$ and ground constructor terms
$(\bterm_1,\ldots,\bterm_m)$ which satisfy the conditions be given.

If $m = 0$, then let $\psi := \bigvee_{((),\varphi) \in A} \varphi$.
By definition of \OK, $\quant{\exists}{\bvar_1 \ldots
\bvar_k}{\psi}$ is valid and $\FV(\psi) \subseteq \{ \avar_1,\ldots,
\avar_n,\bvar_1,\ldots,\bvar_k \}$.  That is, there are values $v_1,
\ldots,v_k$ such that for all values $u_1,\ldots,u_n$ the ground
constraint $\psi[\seq{\bvar}:=\seq{v},\seq{\avar}:=\seq{u}]$ is valid.
In particular, we can take $\seq{\aterm}$ for $\seq{u}$.  Define
$\gamma := [\bvar_1:=v_1,\ldots,\bvar_k:=v_k,\avar_1:=\aterm_1,
\ldots,\avar_n:=\aterm_n]$.
Then $\psi\gamma$ is valid, and since it is ground, some clause in the
disjunction must be valid; so some $((),\varphi) \in A$ where
$\varphi\gamma$ is a valid ground constraint.  This is what the lemma
requires.

If $m > 0$ and $b = \mathsf{either}$ and $\bsort_1$ does not occur in
$\Sigmalogic$, then \OK$(\seq{\avar},A,\mathsf{term})$ holds.
Since there are no values of sort $\bsort_1$, the term $\bterm_1$ is
not a value; for the same reason, variables of sort $\bsort_1$ cannot
occur in any constraint $\varphi$.  Thus, the conditions for the
induction hypothesis are satisfied; we find a suitable $\gamma$ and
$(\seq{\cterm},\varphi)$.

If $m > 0$ and $b = \mathsf{either}$ and $\bsort_1$ does occur in
$\Sigmalogic$, and all constructors with output sort $\bsort_1$ are
values, then \OK$(\seq{\avar},A,\mathsf{value})$ holds;
moreover, $\bterm_1$ is necessarily a value, so we can again apply
the induction hypothesis.

If $m > 0$ and $b = \mathsf{either}$ and there are both values and
other constructors with output sort $\bsort_1$, then both
\OK$(\seq{\avar},T,\mathsf{term})$ and
\OK$(\seq{\avar},V,\mathsf{value})$ must hold.
If $\bterm_1$ is a value, then the conditions to apply the induction
hypothesis with $V$ are satisfied; if not, the conditions to apply it
with $T$ are satisfied!  Since both $T$ and $V$ are subsets of $A$,
this results in a suitable element and substitution.

If $m > 0$ and $b = \mathsf{value}$, then 
we may
assume that $\bterm_1$ is a value.  Let $\avar_{n+1}$ be a fresh
variable of sort $\bsort_1$, and $A' := \{ ((\cterm_2,\ldots,
\cterm_m),\varphi[\cterm_1:=\avar_{n+1}]) \mid (\seq{\cterm},
\varphi) \in A \}$.  
Applying the induction hypothesis with
$n+1,m-1,(\seq{\asort},\bsort_1),(\bsort_2,\ldots,\bsort_m),(\avar_1,
\ldots,\avar_{n+1}),A',\mathsf{either}$ and $(\aterm_1,\ldots,
\aterm_n,\bterm_1)$ and $(\bterm_2,\ldots,\bterm_m)$.
%
gives
an element $(\seq{\cterm},\varphi) \in A'$ and $\gamma$
such that $\gamma(\avar_i) = \aterm_i$ for $1 \leq i \leq n$ and
$\gamma(\avar_{n+1}) = \bterm_1$ and each $\bterm_{i+1} = \cterm_i
\gamma$, and $\varphi\gamma$ is a valid ground logical constraint.
Now, $(\seq{\cterm},\varphi)$ can be written as $((\dterm_2,\ldots,
\dterm_n),\varphi'[\dterm_1:=\avar_{n+1}])$ for some $(\seq{\dterm},
\varphi') \in A$.  So let $\delta$ be the substitution $\gamma \cup
[\dterm_1:=\gamma(\avar_{n+1})]$.  
Noting that by linearity
$\dterm_1$ cannot occur in the other $\dterm_i$, and that
$\varphi'\delta = \varphi\gamma$ because $\avar_{n+1}$ does not occur
in $\varphi'$, 
each further $\delta(\avar_i) =
\gamma(\avar_i) = \aterm_i$ and $\bterm_1 = \gamma(\avar_{n+1}) =
\delta(\dterm_1) = \dterm_1\delta$ and $\bterm_i = \cterm_i\gamma =
\dterm_i\delta$ for larger $i$.\linebreak\vspace{-10pt}

If $m > 0$ and $b = \mathsf{term}$, then we can assume that $\bterm_1$
is not a value.  If all $(\seq{\cterm},\varphi) \in A$ have a variable
for $\cterm_1$, then we use the induction hypothesis and find
a suitable element $((\cterm_2,\ldots,\cterm_m),\varphi) \in A'$ and
substitution $\gamma$ with $\gamma(\avar_i) = \aterm_i$ for all $1
\leq i \leq n$, such that each $\bterm_i = \cterm_i\gamma$ ($i > 1$)
and $\varphi\gamma$ is a valid ground logical constraint.  Choose
$\delta := \gamma \cup [\cterm_1:=\bterm_1]$ (this is safe by
linearity).  The same requirements are satisfied, and also $\bterm_1
= \cterm_1\delta$!

Finally, suppose $m > 0$ and $b = \mathsf{term}$ and $A$ has some
element $(\seq{\cterm},\varphi)$ where $\cterm_1$ is not a variable;
by assumption it is also not a value.
By the conditions, we may assume that always $\cterm_1 \notin
\FV(\varphi)$.
Let $\afun_1,\ldots,\afun_k$ be all constructors in $\Sigmaterms
\setminus \Values$ with output sort $\bsort_1$.  Since also
$\bterm_1$ is not a value, but is a ground constructor term, it can
only have the form $\afun_p(\dterm_1,\ldots,\dterm_k)$ for some $p,
\seq{\dterm}$.
%
Observing that \OK$(\seq{\avar},B_p,\mathsf{either})$ must
hold, we use the induction hypothesis, for $\seq{\avar},\ 
\seq{\aterm}$ and $(\dterm_1,\ldots,\dterm_k,\bterm_2,\ldots,
\bterm_m)$, and find both a suitable tuple $((\eterm_1,\ldots,
\eterm_k,\cterm_2,\ldots,\cterm_m),\varphi) \in B_p$ and a
substitution $\gamma$ which respects $\varphi$, maps $\seq{\avar}$ to
$\seq{\aterm}$ and has $\eterm_i\gamma = \dterm_i$ and
$\cterm_j\gamma = \bterm_j$ for all $i,j$.

By definition of $B_p$, we have $((\cterm_1,\ldots,\cterm_m),
\varphi) \in A$ for some $\cterm_1$ which is either a variable (in
which case all $\eterm_i$ are fresh variables), or $\cterm_1 =
\afun_p(\eterm_1,\ldots,\eterm_k)$.
In the case of a variable, $\cterm_1$ cannot occur in
any of the other $\cterm_i$ by the linearity requirement, nor in
$\varphi$ by the conditions.
Thus, we can safely assume that $\cterm_1$ does not occur in the
domain of $\gamma$, and choose $\delta := \gamma
\cup [\cterm_1:=\afun(\dterm_1,\ldots,\dterm_k)]$.  Then
$\varphi\delta = \varphi\gamma$ is still a valid ground constraint,
each $\aterm_i = \avar_i\gamma = \avar_i\delta$ and for $i > 1$ also
$\bterm_i = \cterm_i\gamma = \cterm_i\delta$.  Finally, $\bterm_1 =
\afun_p(\dterm_1,\ldots,\dterm_k) = \cterm_1\delta$ as required.
In the alternative case that $\cterm_1 = \afun_p(\eterm_1,\ldots,
\eterm_k)$, we observe that $\gamma$ already suffices:
each $\aterm_i = \avar_i\gamma$, for $i > 1$ we have $\bterm_i =
\cterm_i\gamma$ and $\bterm_1 = \afun_p(\dterm_1,\ldots,\dterm_k) =
\afun_p(\eterm_1\gamma,\ldots,\eterm_k\gamma) = \cterm_1\gamma$.
\qed
\end{proof}

With this, we can easily reach our main result:

\begin{theorem}
A left-linear and left-value-free constructor-sound LCTRS with rules
$\Rules$ is quasi-reductive if for all defined and calculation
symbols $\afun$:
\OK$((),A_\afun,\mathsf{either})$ holds, where $A_\afun :=
\{ (\seq{\ell},\varphi) \mid \ell \arrz r\ \constraint{\varphi} \in
\Rules \cup \Rulescalc \wedge \ell = \afun(\seq{\ell}) \}$.
\end{theorem}

\begin{proof}
By Lemma~\ref{lem:minex} it suffices to prove that all terms of the
form $\afun(\aterm_1,\ldots,\aterm_n)$ can be reduced, where $\afun$
is a defined or calculation symbol and all $\aterm_i$ are
constructor terms.  This holds if there is a rule
$\afun(\ell_1,\ldots,\ell_n) \arrz r\ \constraint{\varphi} \in \Rules \cup
\Rules_{\mathtt{calc}}$ and a substitution $\gamma$ such that each
$\aterm_i = \ell_i\gamma$ and $\varphi\gamma$ is a satisfiable
constraint.  By Lemma~\ref{lem:algorithm} (which we can apply because
left-hand sides of rules are linear and value-free), that is exactly
the case if \OK$((),A_\afun,\mathsf{either}) = \mathtt{true}$!
\qed
\end{proof}

\section{Conclusions}

In this paper, we have given an algorithm to prove quasi-reductivity
of LCTRSs, whose core idea is to identify missing cases in the
rules.  Although we needed to impose certain restrictions to use
this algorithm, these restrictions seem very reasonable.

Although we have not proved so here, we believe that our method is not
only sound, but also complete for the class of left-linear,
left-value-free constructor-sound LCTRSs.  We intend to explore this
in future work.

The method presented in this paper has been fully implemented in our
tool \ctrl, which is available at
\begin{center}
\url{http://cl-informatik.uibk.ac.at/software/ctrl/}
\end{center}

\bibliographystyle{ACM-Reference-Format-Journals}
\bibliography{references}

\end{document}